\documentclass{article}

\usepackage{amssymb}
\usepackage{amsthm}

\newtheorem{theorem}{Theorem}

\newcommand{\Aut}{\ensuremath{\mathrm{Aut}}}
\newcommand{\F}{\ensuremath{\mathbb{F}}}

\begin{document}

\title{Two Optimal One-Error-Correcting Codes 
of Length $13$ That Are Not Doubly Shortened Perfect 
Codes}

\author{Patric~R.~J.~\"Osterg{\aa}rd and Olli Pottonen%
\thanks{Olli Pottonen has been supported by the Graduate
School in Electronics, Telecommunication and Automation
and by the Nokia Foundation.
This work has also been supported in part by the
Academy of Finland, Grants no.\ 110196 and 130142.}\\
Department of Communications and Networking\\
Helsinki University of Technology TKK\\
P.O.\ Box 3000, 02015 TKK, Finland\\
{\tt patric.ostergard@tkk.fi, olli.pottonen@tkk.fi}
}

\maketitle

\begin{abstract}
The doubly shortened perfect codes of length 13
are classified utilizing the classification of
perfect codes in [P.R.J.\ \"Osterg{\aa}rd and
O. Pottonen, The perfect binary one-error-correcting 
codes of length 15: Part I---Classification, 
\emph{IEEE Trans.\ Inform.\ Theory}, to appear]; 
there are 117821 such (13,512,3) codes.
By applying a switching operation to those codes, 
two more (13,512,3) 
codes are obtained, which are then not doubly shortened
perfect codes.
\end{abstract}

\section{Introduction}

A \emph{binary code} of length $n$ is a subset of $\F_2^n$, where
$\F_2 = \{0, 1\}$ is the field of two elements. All codes in this
paper are binary. The \emph{Hamming distance} between two codewords 
is the number of coordinates in which they differ.  The \emph{minimum
distance} of a code is the minimum Hamming distance between any two
distinct codewords.  An $(n, M, d)$ code has length $n$, cardinality
$M$ and minimum distance $d$. The function $A(n,d)$ gives the
the maximum integer $M$ for which an $(n,M,d)$ code exists.
An $(n,A(n,d),d)$ code is called \emph{optimal}.

For a code with minimum distance $d$, the balls of radius $r = \lfloor
(d-1)/2 \rfloor$ centered around the codewords are nonintersecting and
such a code is called an \emph{$r$-error-correcting code.} If the
balls cover the entire ambient space, then the code---which is
obviously optimal---is called \emph{perfect}, or \emph{$r$-perfect}. 
All $1$-perfect binary codes
have parameters $(2^m-1, 2^{2^m-m-1}, 3)$ with $m$ arbitrary; linear
such codes are known as \emph{Hamming codes}.  
 
The $1$-perfect binary codes of length $15$ were
recently classified~\cite{prfcodes}. There are 5983 such codes,
up to equivalence. Two codes are \emph{equivalent} if one can
be obtained from the other by adding a vector ${\bf x}$ to all 
codewords and letting a permutation $\pi$ act on the coordinates. 
The \emph{automorphism group} of a code consists
of all such mappings $\pi{\bf x}$ from the code onto 
itself.

\emph{Shortening} a code is the process of removing a coordinate and all
codewords that did not have a specified value in that coordinate.
Best and Brouwer \cite{BB77} used linear programming to prove that 
shortening any $1$-perfect binary code $i$ times with $1 \leq i \leq 3$
yields an optimal $(2^m-1-i, 2^{2^m-m-1-i},3)$ code. But can one
obtain \emph{all} optimal codes
with those parameters, up to equivalence, by 
shortening in that manner?
Etzion and Vardy \cite{EV} asked this question, and
Blackmore \cite{Blackmore99} gave an affirmative answer to the 
question for $i=1$. The main result of the current work is a
negative answer to the question for $i=2$.

The main result of the paper relies on a technique 
for transforming one error-correcting code into another;
this technique, called switching, is considered in
Section \ref{switch}. The classification of doubly shortened
1-perfect codes of length 15 and the result obtained by 
applying switching to these---another two optimal codes---is
discussed in Section \ref{result}. The paper is concluded
in Section \ref{con}.

\section{Switching}

\label{switch}
Switching as a general framework comprises local transformations
of combinatorial structures that keep some of the main parameters of
the structure unchanged. For example, a \mbox{2-switch} of a graph does
not change the degrees of the vertices \cite[p.\ 46]{W}.
In coding theory, switching has in particular been used
to construct new perfect codes from old ones \cite{S}. 
We shall here see
that as a code switch maintains the minimum distance of the
code, it is applicable to any error-correcting codes, not just
perfect ones. The possibility of using switching more generally
for error-correcting codes might seem obvious, but we have not
encountered any related comments in the literature so we include
a comprehensive treatment here.

To switch a binary code with minimum distance $d$, one picks a 
coordinate and forms a graph with one vertex for each codeword 
and an edge between two codewords that are at distance $d$ from
each other \emph{and} that differ in the particularized coordinate.
We call the auxiliary graph obtained in this manner a
\emph{switching graph}. \emph{Switching} now means
changing the value of the particularized coordinate in the
codewords of a connected component of the switching graph.

\begin{theorem}
Switching does not reduce minimum distance.
\end{theorem}

\begin{proof}
Since at most one coordinate value is changed in each codeword
and all changes are carried out in the same coordinate, only
the distance between pairs of words that are originally at
distance $d$ from each other can decrease to $d-1$. But such
pairs of codewords either have the same value or
different values in the 
particularized coordinate. In the former case the distance
cannot decrease, and in the latter case the codewords are
adjacent in the switching graph and belong to the same connected
component (so the switch does not affect the distance).
\end{proof}

With a connected switching graph, switching gives just
an equivalent code, but it is not difficult to come up with
sufficient conditions for the switching graph not to be connected.
For example, if $d$ is odd, then two words at odd distance from
each other when ignoring the particularized coordinate do not
belong to the same connected component. This is in fact the reason why
1-perfect binary codes always have at least two components;
see \cite{S} and its references.

For a given code, switching is a tool for obtaining other
codes. By applying switching to a code in all possible ways---with
respect to both picking the particularized coordinate and the 
connected component of the switching graph---and repeatedly
doing the same for the new codes until no further
codes are found, one obtains 
the \emph{switching class} in which a code resides. 
As an example, switching
partitions the 1-perfect binary codes of length 15 into nine
switching classes \cite{OP2}.

\section{Results}

\label{result}

By shortening the 1-perfect binary codes of length 15 twice in
all possible ways and rejecting equivalent codes, one gets a
classification of doubly shortened 1-perfect codes of length 15.
In this manner, we obtained $117821$ inequivalent $(13,512,3)$ codes
from the 5983 1-perfect $(15,2048,3)$ codes. Detecting equivalent
codes was the main challenge in this endeavour; this was done by
computing canonical equivalence class representatives with an
algorithm from~\cite{prfcodes}.

The $117821$ doubly shortened 1-perfect codes are partitioned into 21
switching classes. In the calculation of switching classes, two more
$(13,512,3)$ codes were encountered. Consequently, these are not
doubly shortened 1-perfect codes. They have automorphism groups of
orders 128 and 96, and both reside in the largest of the switching
classes, which contains $115971$ codes. The large automorphism groups
allows succinct description of the codes, which can be found in
Table~\ref{tbl:codes}. The automorphisms are given as permutations
acting on coordinates, and if a coordinate is marked with an overline,
then the value in that coordinate should be flipped before applying
the permutation. The coordinates are numbered from left to right. 
The codes are also available, in non-compressed form, in the arXiv
source of this document.

\begin{table}
\caption{Two $(13,512,3)$ codes}
\label{tbl:codes}
\begin{tabular}{ll}
\hline\noalign{\smallskip}
First code:\\
\noalign{\smallskip}\hline\noalign{\smallskip}
Automorphism group generators: & \hspace*{43mm}\\
$(1\ 3\ 2\ 13)(\overline{4}\ \overline{7}\ \overline{8}\ 9)
(5\ 10\ 6\ \overline{11})$ &

$(\overline{1}\ 3\ \overline{2}\ 13)(\overline{4}\ 8)(\overline{5})
(\overline{6})(10\ \overline{11})(\overline{12})$ \\

$(3\ 13)(\overline{4}\ \overline{9})(5\ 10)(\overline{6}\ \overline{11})
(\overline{7}\ \overline{8})(\overline{12})$ &

$(\overline{3}\ \overline{13})(4\ 10)(\overline{5}\ \overline{9})
(\overline{6}\ \overline{7})(\overline{8}\ \overline{11})(\overline{12})$ \\

\end{tabular}\\
\begin{tabular}{llll}
\multicolumn{4}{l}{Orbit representatives:}\\
$0000000000000$ &
$1000000010100$ &
$1000011001100$ &
$1010010000100$ \\
\end{tabular}\\
\begin{tabular}{ll}
\noalign{\smallskip}\hline\noalign{\smallskip}
Second code:\\
\noalign{\smallskip}\hline\noalign{\smallskip}
Automorphism group generators: & \hspace*{40.5mm} \\
$(\overline{3}\ 7)(\overline{4}\ \overline{13}\ \overline{6}\ \overline{8})
(5\ \overline{11})(\overline{9})(\overline{10})(\overline{12})$ &

$(4\ 6)(\overline{5})(8\ 13)(\overline{9})(10\ 12)(\overline{11})$ \\

$(1\ \overline{7}\ 3)(\overline{2})(4\ \overline{13}\ \overline{10})
(\overline{5}\ \overline{9}\ \overline{11})(6\ \overline{8}\ \overline{12})$

\end{tabular}\\
\begin{tabular}{lll}
\multicolumn{3}{l}{Orbit representatives:}\\
$0000000000000$ &
$1000000111000$ &
$1010100101000$ \\
$0000001101000$ &
$0010101111000$ &
$1000000001010$ \\
\end{tabular}\\
\begin{tabular}{ll}
\noalign{\smallskip}\hline\noalign{\smallskip}
\hspace*{85.5mm}\\
\end{tabular}
\end{table}

As an independent verification of the fact that the two codes
are not doubly shortened 1-perfect codes, we applied the 
algorithm from \cite{prfcodes}---solving instances of the
exact cover problem with the {\sf libexact}
library~\cite{libexact}---for constructing 1-perfect codes 
from partial codes. 

Further shortening of the two codes reveals---with the
computational approach just mentioned---that they lead to
some $(12, 256, 3)$ and $(11, 128, 3)$ codes that are not,
respectively, triply and four times shortened 1-perfect codes of 
length 15, and also to some codes of length 11 that are
four time shortened 1-perfect codes.

The fact that all new codes found by switching are indeed
non-lengthenable increases the credibility of the classification. To
gain even more confidence in it, the consistency of the results was
verified by counting the number of distinct perfect binary codes of
length $15$ in two different ways. Using the orbit-stabilizer theorem
for the classification of these codes, it was concluded that there are
$1\,397\,746\,513\,516\,953\,600$ such codes~\cite{prfcodes}. We also
know that the number of codes is
\begin{equation}\label{eq:count}
\sum_{C \in \mathcal{C}} \frac{13! \cdot 2^{13} \cdot E(C)}{\Aut{C}}
\end{equation}
where $\mathcal{C}$ contains equivalence class representatives of the
twice shortened perfect codes, $E(C)$ is the number of distinct ways
of extending $C$ to a perfect binary code, and $13! \cdot 2^{13}$ is
the order of the acting group. As this formula yielded the expected
result, the computations are most likely correct.

\section{Conclusion}

\label{con}
The current work settles an open problem but leads to some 
natural further questions that we have so far been unable to
answer. Have all $(13, 512, 3)$ codes now
been found? Do the two counterexamples have 
some particular property that easily shows that they are not 
doubly shortened
perfect codes? Can the structure of the two codes be generalized
or is there some construction that can be applied to them
to obtain an infinite family of similar codes?

As a final note we remark that if we relax the requirement
that the codes in this study be doubly shortened perfect codes of
length 15 and allow shortenings of \emph{any} perfect codes,
then a recent result by 
Avgustinovich and Krotov~\cite{AvKr09} shows that such
shortenings always exist.

\end{document}